\RequirePackage[l2tabu, orthodox]{nag}
\documentclass{llncs}

\usepackage{etex,cite,booktabs,subcaption,amssymb,tikz,multirow,algorithm,algpseudocode,verbatim,setspace,stfloats,xspace, amsthm, expl3}
\usepackage[textsize=small]{todonotes}
\usepackage{microtype}
\usetikzlibrary{shapes}
\usetikzlibrary{calc}

\newtheorem{observation}[theorem]{Observation}

\newcommand{\norm}[1]{\left\Vert#1\right\Vert}

\newcommand{\Real}{\mathbb R}
\newcommand{\eps}{\varepsilon}

\newcommand{\ball}{\mathcal{B}}

\newcommand{\Frechet}{Fr\'echet\xspace}
\newcommand{\ac}{separation corner\xspace}
\newcommand{\acs}{separation corners\xspace}
\newcommand{\cyl}{\mathcal{C}}
\newcommand{\pl}{\mathrm{left}}
\newcommand{\pr}{\mathrm{right}}
\usetikzlibrary{shapes}
\usetikzlibrary{calc}
\newcommand{\tikzsetnextfilename}[1] { }

\newcommand{\eval}[0]  { \pgfmathtruncatemacro }

\makeatletter
\renewcommand{\todo}[2][]{\tikzexternaldisable\@todo[#1]{#2}\tikzexternalenable}
\newcommand{\getxy}[3]{%
  \tikz@scan@one@point\pgfutil@firstofone#1\relax
  \edef#2{\the\pgf@x}%
  \edef#3{\the\pgf@y}%
}
\makeatother

\makeatletter
\pgfmathdeclarefunction{atan3}{2}{%
	\begingroup
		\pgfmathsetmacro \pgfmathx {abs(#1)}%
		\pgfmathsetmacro \pgfmathy {abs(#2)}%
		\ifdim \pgfmathx pt>\pgfmathy pt\relax%
			\pgfmathparse{atan2(#1,#2)}%
		\else
			\pgfmathparse{90 - atan2(#2,#1)}%
		\fi%
 		\pgfmath@smuggleone\pgfmathresult%
	\endgroup
}
\makeatother

\newcommand{\ceps}{0.2cm}

\newcommand{\tikzdefines}[0] {
    \tikzstyle{cblack}=[circle, draw, thick, solid, fill=black, scale=.15]
    \tikzstyle{cblue}=[circle, draw, solid, fill=cyan!20, scale=.4]
    \tikzstyle{cred}=[red, circle, draw, solid, fill=red!20, scale=.4]
    \tikzstyle{cgreen}=[circle, draw, solid, fill=green!60, scale=.4]
    \tikzstyle{cdotted}=[circle, draw, thick, dotted, fill=none, scale=1*\ceps]
    \tikzstyle{lblue}=[double distance=.1cm,line cap=round,double=cyan!20]
}
\pgfdeclarelayer{background2}
\pgfdeclarelayer{background}
\pgfsetlayers{background2,background,main}

\newcommand{\xcorner}[5] {
    \coordinate (a1) at ($(#3)!2*\ceps!270:(#2)$);
    \coordinate (a2) at ($(a1) + ($(0,0)!\ceps!($(#3)-(#2)$)$)$);
    \coordinate (a3) at ($(#4)!2*\ceps!90:(#5)$);
    \coordinate (a4) at ($(a3)+($(0,0)!\ceps!($(#4)-(#5)$)$)$);
    \coordinate (#1_b) at (intersection of a1--a2 and a3--a4);
    \coordinate (a1) at ($(#1_b)+($(0,0)!\ceps!($(#3)-(#2)$)$)$);
    \coordinate (a2) at ($(#1_b)+($(0,0)!\ceps!($(#5)-(#4)$)$)$);
    \coordinate (#1_a) at (intersection of #2--#3 and a2--#1_b);
    \coordinate (#1_c) at (intersection of #4--#5 and a1--#1_b);
    \coordinate (#1_d) at ($(#1_a)!\ceps!270:(#1_b)$);
    \coordinate (#1_e) at ($(#1_c)!\ceps!90:(#1_b)$);

    \node[cblack] () at (#1_a) {};
    \node[cblack] () at (#1_b) {};
    \node[cblack] () at (#1_c) {};
}

\newcommand{\acorner}[5] {
    \xcorner{#1}{#2}{#3}{#4}{#5}
    \node[cblue] (#1_nd) at (#1_d) {};
}

\newcommand{\bcorner}[5] {
    \xcorner{#1}{#2}{#3}{#4}{#5}
    \node[cblue] (#1_ne) at (#1_e) {};
}

\newcommand{\ucorner}[5] {
    \xcorner{#1}{#2}{#3}{#4}{#5}
    \node[cblue] (#1_ne) at (#1_e) {};
    \node[cblue] (#1_nd) at (#1_d) {};
}

\newcommand{\pathborder}[2]{
    [
        create hullcoords/.code={
            \global\edef\namelist{#1}
            \foreach [count=\counter] \nodename in \namelist {
                \global\edef\numberofnodes{\counter}
                \coordinate (hullcoord\counter) at (\nodename);
            }
            \pgfmathtruncatemacro \numberofnodes {\numberofnodes - 1}
            \foreach [count=\counter] \nodenum in {\numberofnodes,...,1} {
                \pgfmathtruncatemacro \next {\counter + \numberofnodes + 1}
                \coordinate (hullcoord\next) at (hullcoord\nodenum);
            }
            \pgfmathtruncatemacro \numberofnodes {\numberofnodes * 2 + 1}
            \coordinate (hullcoord0) at (hullcoord2);
            \pgfmathtruncatemacro \lastnumber {\numberofnodes+1}
            \coordinate (hullcoord\lastnumber) at (hullcoord2);
        },
        create hullcoords
    ]
    ($(hullcoord1)!#2!90:(hullcoord0)$) node[circle] {}
    \foreach [
        evaluate=\currentnode as \prevnode using \currentnode-1,
        evaluate=\currentnode as \nextnode using \currentnode+1
    ] \currentnode in {2,...,\numberofnodes} {
        let
            \p1 = ($(hullcoord\currentnode) - (hullcoord\prevnode)$),
            \n1 = {atan3(\x1,\y1) + 90pt},
            \p2 = ($(hullcoord\nextnode) - (hullcoord\currentnode)$),
            \n2 = {atan3(\x2,\y2) + 90pt},
            \n3 = {Mod(\n2-\n1,360) - 360},
            \n4 = {cos(.5*(\n3+360))},
            \n5 = {\n1+.5*\n3}
        in
        {
            \ifdim \n3 < -180.05pt {
                -- ($(hullcoord\currentnode) + (\n5:-1/\n4*#2)$)
            } \else {
                -- ($(hullcoord\currentnode)!#2!-90:(hullcoord\prevnode)$)
                arc [start angle=\n1, delta angle=\n3, radius=#2]
            } \fi
        }
    }
}

\ExplSyntaxOn
\newcommand\latinabbrev[1]{
  \peek_meaning:NTF . {% Same as \@ifnextchar
    #1\@}%
  { \peek_catcode:NTF a {% Check whether next char has same catcode as \'a, i.e., is a letter
      #1.\@ }%
    {#1.\@}}}
\ExplSyntaxOff

\renewenvironment{thebibliography}[1]{%
\begin{oldthebibliography}{#1}%
\setlength{\itemsep}{0ex}%
}%
{%
\end{oldthebibliography}%
}

\title{Matching Curves to Imprecise Point Sets\\using \Frechet Distance}
\institute{
Dept. of Computer \& Info. Sci. \& Eng.\\
University of Florida, Gainesville, FL, USA\\
{\tt \{accisano, ungor\}@cise.ufl.edu}
}
\author{Paul Accisano \and Alper {\"{U}ng\"{o}r}}

\begin{document}

\maketitle

\begin{abstract}
    Let $P$ be a polygonal curve in $\mathbb{R}^d$ of length $n$, and $S$ be a point-set of size $k$.  The Curve/Point Set Matching problem consists of finding a polygonal curve $Q$ on $S$ such that the \Frechet distance from $P$ is less than a given $\eps$.  We consider eight variations of the problem based on the distance metric used and the omittability or repeatability of the points.  We provide closure to a recent series of complexity results for the case where $S$ consists of precise points.  More importantly, we formulate a more realistic version of the problem that takes into account measurement errors.  This new problem is posed as the matching of a given curve to a set of imprecise points.  We prove that all three variations of the problem that are in P when $S$ consists of precise points become NP-complete when $S$ consists of imprecise points.  We also discuss approximation results.
    
\end{abstract}

\section{Introduction}
Matching a curve and a set of points is a typical geometric problem that arises in many science and engineering fields, such as computer aided design, computer graphics and vision, and protein structure prediction \cite{Brakatsoulas05,Jiang08}. In these applications, data is typically gathered as a point set through a scanner, and the goal is often to find certain objects, described as polygonal curves, in the scene.  In order to perform a matching, one needs a similarity metric between geometric constructs.

In this paper, we study the problem of curve and point set matching, using the \Frechet distance as the similarity metric.  Given a point set and a polygonal curve, the goal is to connect the points into a new polygonal curve that is similar to the given curve.  Formally, given a polygonal curve $P$ of length $n$, a point set $S$ of size $k$, and a real number $\eps > 0$, determine whether there exists a polygonal curve $Q$ on a subset of the points of $S$ such that $\delta_F(P, Q) \le \eps$.

The version of this problem in which the points to be matched are precise has been well studied in the literature \cite{Accisano12,Maheshwari11,Wylie12}, and we refer to it in this paper as the \textbf{\emph{Curve/Point Set Matching (CPSM)}} problem.  However, the limitations of modern scanner technology suggest that a more realistic version of this problem would be to consider the input points as imprecise regions.  In this paper, we introduce this new version of the problem and refer to it as the \textbf{\emph{Curve/Imprecise Point Set Matching (CIPSM)}} problem.

\begin{table}[h]
    \centering
    \begin{tabular}{l@{\hspace{1em}}l@{\hspace{1em}}|@{\hspace{3em}}cl@{\hspace{3em}}c@{\hspace{-1em}}c}
        \toprule
                &               &   \multicolumn{2}{l}{Discrete} & Continuous\\
        \midrule
        Subset  & Unique        & NP-C & \cite{Wylie12}     & \textbf{NP-C}  \\
                & Non-Unique    & P    & \cite{Wylie12}     & P & \cite{Maheshwari11} \\
        All-Points & Unique        & NP-C & \cite{Wylie12}     & NP-C & \cite{Accisano12}\\
                & Non-Unique    & P    & \cite{Wylie12}     & NP-C & \cite{Accisano12} \\
        \bottomrule
    \end{tabular}%
    \bigskip
    \caption{Eight versions of the CPSM problem and their complexity classes. \vspace{-10pt}}
    \label{tab:results}%
\end{table}%

Eight versions of the original CPSM problem can be classified based on whether the use of all points is enforced, whether points are allowed to be visited more than once, and whether the \Frechet distance metric used is discrete or continuous.  Table \ref{tab:results} summarizes the versions and their complexity classes.

\textbf{Our results.} At an earlier workshop, we presented preliminary results showing that the CPSM problem is NP-complete when coverage of all points is enforced \cite{Accisano12}.  Here we extend this work by also proving the last remaining open question in Table \ref{tab:results}, the Unique Subset version (bold), is also NP-complete (Section \ref{sec:precise}).  More importantly, to the best of our knowledge, this stands as the first paper to formulate and present complexity results on matching a curve to imprecise points using \Frechet distance (CIPSM).  Naturally, all the versions shown in Table \ref{tab:results} that are NP-complete are also NP-complete for their imprecise variations.  However, we show that the other three versions are also NP-complete, using nontrivial reductions (Sections \ref{sec:imprecise}, \ref{sec:dimprecise}).  Two of the three are shown to remain NP-complete even if the given curve is simple. Complexity results on the CIPSM problem are summarized in Table \ref{tab:impresults} in section \ref{sec:imprecise}.  Finally, we briefly present some simple approximation results.

\section{Previous Work}
The basic \Frechet distance problem asks, given two geometric objects and a real number $\eps > 0$, is the \Frechet distance $\delta_F$ between the two objects less than $\eps$?  It was shown by Alt and Godau \cite{Alt95} that, when the objects in question are polygonal curves of length $n$ and $m$, this problem can be solved in $O(nm)$ time.  They also showed that finding the exact \Frechet distance between the two curves can be done in $O(nm \log(nm))$ time.

Alt et al.\ \cite{Alt03} examined the following variant of the \Frechet distance problem.  Given a polygonal curve $P$ of length $n$, a graph $G$ of complexity $m$, and a real number $\eps > 0$, determine whether there exists a path in $G$ that stays within \Frechet distance $\eps$ of $P$.  They presented an algorithm that decides this problem in time $O(nm\log(m))$.  When the input graph is a clique, their problem becomes the Continuous Non-unique Subset version of the CPSM problem.  Maheshwari et al.\ presented an algorithm in \cite{Maheshwari11} that decides this problem in time $O(nk^2)$, improving on the result in \cite{Alt03} by a log factor.  They also showed that the curve of minimal \Frechet distance can be computed in time $O(nk^2\log(nk))$ using parametric search.

Wylie and Zhu \cite{Wylie12} also explored the CPSM problem from the perspective of \emph{discrete} \Frechet distance, which only takes into account the distance at the vertices along the curves.  They showed that the non-unique versions were solvable in $O(nk)$ time, and the unique versions were NP-complete, as listed in Table \ref{tab:results}.

A typical scenario in geometric applications occurs when there exist measurement errors or finite precision computations.  In such cases, it makes sense to integrate data imprecision into the formulation of the geometric problem \cite{Ahn12b, Loffler09}.  Related to our work, Ahn et al.\ \cite{Ahn12b} recently studied discrete \Frechet distance between two polygonal chains with imprecise vertices.  Even when we limit ourselves to discrete \Frechet distance, this differs from our problem where only one curve is given.  The CIPSM problem turns out to be hard, while their version of the problem admits a polynomial time solution.

\section{Preliminaries} \label{sect:hardness_prelims}
Below, we present the notation that will be used throughout the paper, some of which is similar to the notation used by earlier work \cite{Accisano12, Alt03,Maheshwari11}.  More will be introduced later as needed.  Given two curves $P, Q : [0, 1] \rightarrow \mathbb{R}^d$, the \emph{\Frechet distance} between $P$ and $Q$ is defined as $\delta_F(P, Q) = \inf_{\sigma, \tau} \max_{t\in[0,1]} \norm{P(\sigma(t)), Q(\tau(t))}$, where $\sigma, \tau : [0, 1] \rightarrow [0, 1]$ range over all continuous non-decreasing surjective functions \cite{Ewing85}.

Let the continuous function $P : [0, 1] \rightarrow \mathbb{R}^d$ represent a curve in $\mathbb{R}^d$.  Given two points $u, v \in P$, we use the notation $u \prec v$ if $u$ occurs before $v$ on a traversal of $P$.  The relation $\succ$ is defined analogously.  For a subcurve $R \subseteq P$, we denote the first and last point of $R$ along $P$ as left$(R)$ and right$(R)$, respectively.

For a given point $p \in \mathbb{R}^d$ and a real number $\eps > 0$, let $\ball(p, \eps) \equiv \{q \in \mathbb{R}^d : \norm{pq} \le \eps$ denote the \emph{ball} of radius $\eps$ centered at $p$, where $\norm{\cdot}$ denotes Euclidean distance.  Given a line segment $L \subset \mathbb{R}^d$, let $\mathcal{C}(L, \eps) \equiv \bigcup_{p \in L} \mathcal{B}(p, \eps)$ denote the \emph{cylinder} of radius $\eps$ around $L$.  Note that a necessary condition for two polygonal curves $P$ and $Q$ to have \Frechet distance less than $\eps$ is that the vertices of $Q$ must all lie within the cylinder of some segment of $P$.

\section{Continuous Unique Subset CPSM Complexity} \label{sec:precise}
We now show that the Continuous Unique Subset version of the CPSM problem is NP-complete.  Our reduction is from the (3,B2)-SAT problem, a variant of the famous 3-SAT problem in which each literal, positive and negative, is restricted to occur exactly twice.  This variant was shown to be NP-complete in \cite{Berman03}.

Let $\Phi$ be a formula given as input to the (3,B2)-SAT problem.  We construct a polygonal curve $P$ and a point set $S$ such that $\Phi$ is satisfiable if and only if there exists a vertex-unique polygonal curve $Q$ with \Frechet distance at most $\eps$ from $P$.  A curve is \emph{vertex-unique} if it has no shared vertices.

Our reduction makes use of a gadget we call \emph{\acs}, a special version of which was introduced in \cite{Accisano12}.  These corner constructs force a choice between two possible paths in $S$, allowing the effects of binary choice to be propagated to other parts of the construction.  (Figure \ref{fig:switch} Left)

%\begin{figure*}[ht]
%    \begin{subfigure}[t]{0.45\textwidth}
%	    \centering
%        \input{fig_separation1.tex}
%        \label{fig:separation1}
%    \end{subfigure}
%    \hfill
%    \begin{subfigure}[t]{0.45\textwidth}
%    	\centering
%        \input{fig_separation2.tex}
%        \label{fig:separation2}
%    \end{subfigure}
%    \{The two dashed curves are the only possible curves on $S$ with \Frechet distance less than 1 from the given (solid) curve.}
%    \label{fig:separation}
%\end{figure*}

\begin{figure}[t]
    \begin{subfigure}[t]{0.45\textwidth}
        \tikzsetnextfilename{switch1}
\begin{tikzpicture}[scale=.3]
    \tikzdefines

    \foreach [count=\x] \pt in {(0,5), (0,2), (2,2), (2,0), (14,0), (14,2), (16, 2), (16,5)}
    	\node[cblack] (p\x) at \pt {};

    \draw ($(p1)$) \foreach \x in {2,...,8} {--($(p\x)$)};
    \draw[dotted] \pathborder{p1,p2,p3,p4,p5,p6,p7,p8}{1cm};

    \node[cblue](s1) at (1,5) {};
    \node[cblue](s2) at (0,1) {};
    \node[cblue](s3) at (1,0) {};
    \node[cblue](s4) at (15,0) {};
    \node[cblue](s5) at (16,1) {};
    \node[cblue](s6) at (15,5) {};

    \draw[thick, red, dashed] (s1) -- (s2) -- (s5) -- (s6);
    \draw[thick, red, dashed] (s1) -- (s3) -- (s4) -- (s6);

\end{tikzpicture}
        %\bigskip
	\label{fig:switch1}
    \end{subfigure}
    \hfill
    \begin{subfigure}[t]{0.45\textwidth}
        \tikzsetnextfilename{switch2}
\begin{tikzpicture}[scale=.3]
    \tikzdefines

    \foreach [count=\x] \pt in {(0,5), (0,2), (2,2), (2,0), (14,0), (14,2), (16, 2), (16,5)}
    	\node[cblack] (p\x) at \pt {};

    \draw ($(p1)$) \foreach \x in {2,...,8} {--($(p\x)$)};
    \draw[dotted] \pathborder{p1,p2,p3,p4,p5,p6,p7,p8}{1cm};

    \node[cblue](s1) at (1,5) {};
    \node[cblue](s2) at (0,1) {};
    \node[cblue](s3) at (1,0) {};
    \node[cblue](s4) at (15,0) {};
    \node[cblue](s5) at (16,1) {};
    \node[cblue](s6) at (15,5) {};

    \node[cblue] (i1) at ($(s2)!.5!(s5)$) {};

    \draw[thick, red] (s1) -- (s2) -- (i1) -- (s4) -- (s6);

\end{tikzpicture}
        \label{fig:switch2}
    \end{subfigure}
    \caption{Normally, once the curve starts down one path, changing to the other is impossible.  An extra point on the cylinder boundary allows the curve to start on one path and switch to the other midway. }
     \label{fig:switch}
\end{figure}
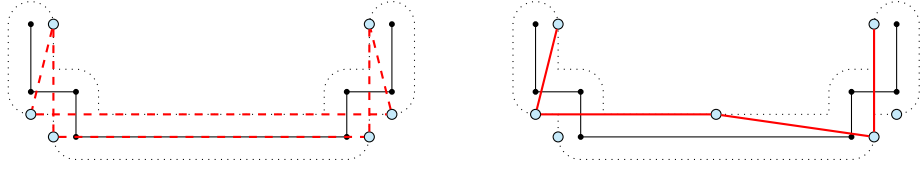

We will first create a series of small chains, each consisting of two \acs, laid out horizontally.  These chains will represent the variables of $\Phi$, and the four corner points used in the \acs will represent the four literal instances of the variable.  Then, we will create a \ac \emph{loop} for each clause.  However, instead of allowing both possible paths, we will force one of the two to be chosen.  At the end of the loop, we will force the chosen path to terminate in a dead end.  The loop will be arranged so that the literal points corresponding to the literals used in the clause provide an opportunity for the curve to ``change tracks'' and avoid the dead end.  Since points cannot be used more than once, a literal point will only be available for use to change tracks if it was not already used in the initial variable assignment path.  Thus, there exists a path that can traverse the entire curve if and only if $\Phi$ has a satisfying assignment.

Extra points on the cylinder boundaries are necessary for the curve to switch between paths.  (Figure \ref{fig:switch} Right).  Ordinarily, there exist only two path possibilities, and once the first corner point is decided, the curve is fully determined until the end of the loop.  However, an extra point on the cylinder boundary allows the curve to change tracks to the other path possibility.  We will use this property when constructing the clause section of the construction.

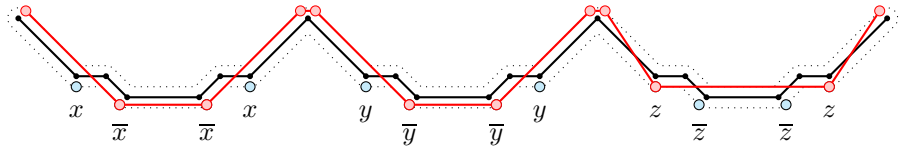
\begin{figure}[b]
    \centering
    \makebox[\textwidth][c]{
    \tikzsetnextfilename{variables}
\begin{tikzpicture}[scale=.7]
    \tikzdefines

    \node[cblack] (a0_p) at (-0.5,1.5) {};
    \foreach [count=\y] \x in {0,...,2} {
        \ucorner{a\x}{0+5.5*\x,1}{1+5.5*\x,0}{1+5.5*\x,0}{2+5.5*\x,0};
        \ucorner{b\x}{1+5.5*\x,0}{2+5.5*\x,0}{3.5+5.5*\x,0}{4.5+5.5*\x,1};
        \node[cblack] (a\y_p) at (5+5.5*\x,1.5) {};

        \node[cred] (a\x_q) at ($(a\x_p)!\ceps!90:(a\x_a)$) {};
        \node[cred] (b\x_q) at ($(a\y_p)!\ceps!270:(b\x_c)$) {};
    };

    \foreach [count=\y] \x in {0,...,2} {
            \draw[thick] (a\x_p)--(a\x_a)--(a\x_b)--(a\x_c)--(b\x_a)--(b\x_b)--(b\x_c)--(a\y_p);
    }

%    \node () at ($($(a0_b)!.5!(b0_b)$)+(0,.5)$) {$x$};
%    \node () at ($($(a1_b)!.5!(b1_b)$)+(0,.5)$) {$y$};
%    \node () at ($($(a2_b)!.5!(b2_b)$)+(0,.5)$) {$z$};

    \node[label={below:$x$}] () at (a0_d) {};
    \node[label={below:$x$}] () at (b0_e) {};
    \node[label={below:$\overline{x}$}] () at (a0_e) {};
    \node[label={below:$\overline{x}$}] () at (b0_d) {};
    \node[label={below:$y$}] () at (a1_d) {};
    \node[label={below:$y$}] () at (b1_e) {};
    \node[label={below:$\overline{y}$}] () at (a1_e) {};
    \node[label={below:$\overline{y}$}] () at (b1_d) {};
    \node[label={below:$z$}] () at (a2_d) {};
    \node[label={below:$z$}] () at (b2_e) {};
    \node[label={below:$\overline{z}$}] () at (a2_e) {};
    \node[label={below:$\overline{z}$}] () at (b2_d) {};

    \draw[thick, red] (a0_q)
                        --(a0_e) node[cred, thin]{}
                        --(b0_d) node[cred, thin]{}
                        --(b0_q) --(a1_q)
                        --(a1_e) node[cred, thin]{}
                        --(b1_d) node[cred, thin]{}
                        --(b1_q) --(a2_q)
                        --(a2_d) node[cred, thin]{}
                        --(b2_e) node[cred, thin]{}
                        --(b2_q);

    \begin{pgfonlayer}{background}
        \draw[black, dotted] \pathborder{a0_p,a0_a,a0_b,a0_c,b0_a,b0_b,b0_c,a1_p,a1_a,a1_b,a1_c,b1_a,b1_b,b1_c,a2_p,a2_a,a2_b,a2_c,b2_a,b2_b,b2_c,a3_p}{\ceps};
    \end{pgfonlayer}

\end{tikzpicture} 
    }
    \caption{The variable section of the construction for a formula with three variables.  The curve corresponding to the assignment {\sc true}, {\sc true}, {\sc false} is shown.}
    \label{fig:variables}
\end{figure}

\subsection{Construction: Variable Section}
The construction is composed of two sections: one for the variables and one for the clauses.  We begin with the variable section.  The construction starts with a set of \acs for each variable laid out as in Figure \ref{fig:variables}, creating two path alternatives for each variable.  The two path possibilities will correspond to {\sc true} or {\sc false} assignments to that variable.  Note that in order to traverse this part of the construction, either the inner or outer corner points \emph{must} be visited.  We refer to these points as literal points, as they will represent the literals of $\Phi$.  The outer corner points of each variable construct will be referred to as the positive-points, and the inner corner points will be the negative-points.

The purpose of the variable section is to ``use up'' the literal points corresponding to whichever {\sc true}/{\sc false} value is \emph{not} assigned to the variable, leaving the points corresponding to the actual variable value for later use by the clause section.  Figure \ref{fig:variables} shows how an assignment to the variables of $\Phi$ maps to a traversal of the variable section in the construction.  Variables assigned to {\sc true} take the inner path, leaving the outer points available for use later, while variables assigned to {\sc false} take the outer path, leaving the inner points available.

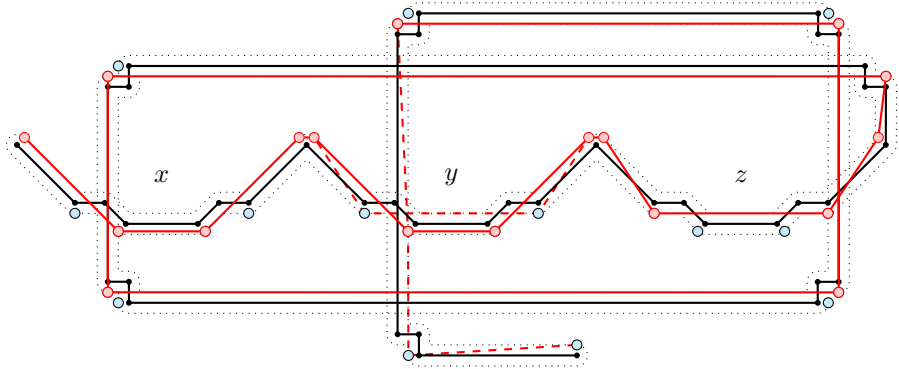
\begin{figure}[b]
    \centering
    \makebox[\textwidth][c]{
    \tikzsetnextfilename{clause}
\begin{tikzpicture}[scale=.7]
    \tikzdefines

    \node[cblack] (a0_p) at (-0.5,1.5) {};
    \foreach [count=\y] \x in {0,...,2} {
        \ucorner{a\x}{0+5.5*\x,1}{1+5.5*\x,0}{1+5.5*\x,0}{2+5.5*\x,0};
        \ucorner{b\x}{1+5.5*\x,0}{2+5.5*\x,0}{3.5+5.5*\x,0}{4.5+5.5*\x,1};
        \node[cblack] (a\y_p) at (5+5.5*\x,1.5) {};

        \node[cred] (a\x_q) at ($(a\x_p)!\ceps!90:(a\x_a)$) {};
        \node[cred] (b\x_q) at ($(a\y_p)!\ceps!270:(b\x_c)$) {};
    };

    \coordinate (x) at ($(a3_p)+(0,1)$);
    \acorner{c1}{a3_p}{x}{1,3}{0,3};

    \node () at ($($(a0_b)!.5!(b0_b)$)+(0,.5)$) {$x$};
    \node () at ($($(a1_b)!.5!(b1_b)$)+(0,.5)$) {$y$};
    \node () at ($($(a2_b)!.5!(b2_b)$)+(0,.5)$) {$z$};

    \foreach [count=\x] \var in {a0_e, b2_e, a1_e} {
        \eval \xa {\x*2-1}
        \eval \xb {\x*2}
        \eval \xc {\x*2+1}
        \eval \dir{pow(-1,\x)}
        \eval \m  {mod(\xc,12)}
        \pgfmathsetmacro \ha {1.75 + floor((\x+1)/2)}
        \pgfmathsetmacro \hb {1.25 + \ha * \dir}

        \coordinate (xd) at ($(c\xa_c)-(\dir,0)$);
        \getxy{($(\var)+(\ceps*\dir, 0)$)}{\cx}{\cy};
        \ucorner{c\xb}{xd}{c\xa_c}{\cx,0}{\cx,\dir};

        \coordinate (xd) at ($(c\xb_c)-(0,\dir)$);
    	\ifnum\m=1
            \acorner{c\xc}{xd}{c\xb_c}{\dir,\hb}{0,\hb};
    	\else
    		\ifnum\m=7
                \bcorner{c\xc}{xd}{c\xb_c}{\dir,\hb}{0,\hb};
    		\else
                \ucorner{c\xc}{xd}{c\xb_c}{\dir,\hb}{0,\hb};
    		\fi
    	\fi
    };

    \node[cblack] (c8_a) at ($(c7_c)+(3,0)$) {};
    \node[cblue] (c8_q) at ($(c8_a)+(0,\ceps)$) {};

    \foreach [count=\y] \x in {0,...,2} {
            \draw[thick] (a\x_p)--(a\x_a)--(a\x_b)--(a\x_c)--(b\x_a)--(b\x_b)--(b\x_c)--(a\y_p);
    }
    \draw[thick] (a3_p)--(c1_a);
    \foreach [count=\y from 2] \x in {1,...,7} {
            \draw[thick] (c\x_a)--(c\x_b)--(c\x_c)--(c\y_a);
    }

    \draw[thick, red] (a0_q)
                        --(a0_e) node[cred, thin]{}
                        --(b0_d) node[cred, thin]{}
                        --(b0_q) --(a1_q)
                        --(a1_e) node[cred, thin]{}
                        --(b1_d) node[cred, thin]{}
                        --(b1_q) --(a2_q)
                        --(a2_d) node[cred, thin]{}
                        --(b2_e) node[cred, thin]{}
                        --(b2_q)
                        --(c1_d) node[cred, thin]{}
                        --(c2_e) node[cred, thin]{}
                        --(c3_d) node[cred, thin]{}
                        --(c4_e) node[cred, thin]{}
                        --(c5_d) node[cred, thin]{}
                        --(c6_e) node[cred, thin]{}
                        ;
    \begin{pgfonlayer}{background}
        \draw[thick, red, dashed] (a1_q)
                        --(a1_d)
                        --(b1_e)
                        --(b1_q);

        \draw[thick, red, dashed] (c6_e) -- (a1_e) -- (c7_e) -- (c8_q);

        \draw[black, dotted] \pathborder{a0_p,a0_a,a0_b,a0_c,b0_a,b0_b,b0_c,a1_p,a1_a,a1_b,a1_c,b1_a,b1_b,b1_c,a2_p,a2_a,a2_b,a2_c,b2_a,b2_b,b2_c,a3_p,    c1_a,c1_b,c1_c,c2_a,c2_b,c2_c,c3_a,c3_b,c3_c,c4_a,c4_b,c4_c,c5_a,c5_b,c5_c,c6_a,c6_b,c6_c,c7_a,c7_b,c7_c,c8_a}{\ceps};
    \end{pgfonlayer}

\end{tikzpicture} 
    }
    \caption{A clause loop for the clause $(\overline{x} \vee \overline{y} \vee z)$.}
    \label{fig:clause}
\end{figure}

\subsection{Construction: Clause Section}
We next create the clause section of the construction, appending it to the variable section.  We begin by adding a \ac loop.  However, we leave out one of the two corner points in the first \ac.  This will force the curve to pick a specific possibility and remove the option to pick the other.  Next, we place more \acs, arranging the loop so that the three literal points corresponding to the clause's literals are exactly on the cylinder boundaries of the segments.  Once this is done, we remove another point from the next \ac in the loop corresponding to the path that was forced earlier, creating a dead end.  The only way to proceed will be to use one of the literal points to change tracks before the dead end is reached.  If no literal point is available for use, then the clause is not satisfied and there will be no way for the curve to proceed while maintaining the appropriate \Frechet distance.

Figure \ref{fig:clause} demonstrates a single clause loop.  Note that the first and last \acs of the clause loop are missing a corner point.  Since the corresponding literal points are already used, there is no place to switch, and the solid curve cannot continue because of the missing point at the end.  However, if the value of variable $y$ is changed from {\sc true} to {\sc false}, the corresponding literal point is free to be used by the clause loop and escape the dead end.  The dashed curve shows this configuration.

This process is then repeated for every clause, with a dead end \ac between each clause loop.  The full construction is therefore only traversable if every clause loop has a point at which it can switch tracks, which corresponds to a satisfying assignment.  Figure \ref{fig:unique} shows a completed construction.

\begin{figure}[h]
    \centering
    \makebox[\textwidth][c]{
    \tikzsetnextfilename{unique}
\begin{tikzpicture}[scale=.7]
    \tikzdefines

    \node[cblack] (a0_p) at (-0.5,1.5) {};
    \foreach [count=\y] \x in {0,...,2} {
        \ucorner{a\x}{0+5.5*\x,1}{1+5.5*\x,0}{1+5.5*\x,0}{2+5.5*\x,0};
        \ucorner{b\x}{1+5.5*\x,0}{2+5.5*\x,0}{3.5+5.5*\x,0}{4.5+5.5*\x,1};
        \node[cblack] (a\y_p) at (5+5.5*\x,1.5) {};

        \node[cblue] (a\x_q) at ($(a\x_p)!\ceps!90:(a\x_a)$) {};
        \node[cblue] (b\x_q) at ($(a\y_p)!\ceps!270:(b\x_c)$) {};
    };

    \coordinate (x) at ($(a3_p)+(0,1)$);
    \acorner{c1}{a3_p}{x}{1,2.8}{0,2.8};

    \node () at ($($(a0_b)!.5!(b0_b)$)+(0,.5)$) {$x$};
    \node () at ($($(a1_b)!.5!(b1_b)$)+(0,.5)$) {$y$};
    \node () at ($($(a2_b)!.5!(b2_b)$)+(0,.5)$) {$z$};

    \foreach [count=\x] \var in {a0_d, b2_e, a1_d,   a2_e, a0_e, b1_e,   b0_d, a2_d, a1_e,   b2_d, b0_e, b1_d} {
        \eval \xa {\x*2-1}
        \eval \xb {\x*2}
        \eval \xc {\x*2+1}
        \eval \dir{pow(-1,\x)}
        \eval \m  {mod(\xc,12)}
        \pgfmathsetmacro \ha {1.75 + floor((\x+1)/2)}
        \pgfmathsetmacro \hb {1.25 + .9 * \ha * \dir}

        \coordinate (xd) at ($(c\xa_c)-(\dir,0)$);
        \getxy{($(\var)+(\ceps*\dir, 0)$)}{\cx}{\cy};
        \ucorner{c\xb}{xd}{c\xa_c}{\cx,0}{\cx,\dir};

        \coordinate (xd) at ($(c\xb_c)-(0,\dir)$);
    	\ifnum\m=1
            \acorner{c\xc}{xd}{c\xb_c}{\dir,\hb}{0,\hb};
    	\else
    		\ifnum\m=7
                \bcorner{c\xc}{xd}{c\xb_c}{\dir,\hb}{0,\hb};
    		\else
                \ucorner{c\xc}{xd}{c\xb_c}{\dir,\hb}{0,\hb};
    		\fi
    	\fi
    };

    \node[cblack] (c26_a) at ($(c25_c)-(3,0)$) {};
    \node[cblue] (c26_q) at ($(c26_a)-(0,\ceps)$) {};

    \foreach [count=\y] \x in {0,...,2} {
            \draw[thick] (a\x_p)--(a\x_a)--(a\x_b)--(a\x_c)--(b\x_a)--(b\x_b)--(b\x_c)--(a\y_p);
    }
    \draw[thick] (a3_p)--(c1_a);
    \foreach [count=\y from 2] \x in {1,...,25} {
            \draw[thick] (c\x_a)--(c\x_b)--(c\x_c)--(c\y_a);
    }

%
%    \draw[thick, red] (a0_q)
%                        --(a0_e) node[cred, thin]{}
%                        --(b0_d) node[cred, thin]{}
%                        --(b0_q) --(a1_q)
%                        --(a1_e) node[cred, thin]{}
%                        --(b1_d) node[cred, thin]{}
%                        --(b1_q) --(a2_q)
%                        --(a2_d) node[cred, thin]{}
%                        --(b2_e) node[cred, thin]{}
%                        --(b2_q)
%                        --(c1_d) node[cred, thin]{}
%                        --(c2_e) node[cred, thin]{}
%                        --(c3_d) node[cred, thin]{}
%                        --(c4_e) node[cred, thin]{}
%                        --(c5_d) node[cred, thin]{}
%                        --(c6_e) node[cred, thin]{}
%                        ;
    \begin{pgfonlayer}{background}
%        \draw[thick, red, dashed] (a1_q)
%                        --(a1_d)
%                        --(b1_e)
%                        --(b1_q);
%
%        \draw[thick, red, dashed] (c6_e) -- (a1_e) -- (c7_e) -- (c8_q);

        \draw[black, dotted] \pathborder{a0_p,a0_a,a0_b,a0_c,b0_a,b0_b,b0_c,a1_p,a1_a,a1_b,a1_c,b1_a,b1_b,b1_c,a2_p,a2_a,a2_b,a2_c,b2_a,b2_b,b2_c,a3_p, c1_a,c1_b,c1_c,c2_a,c2_b,c2_c,c3_a,c3_b,c3_c,c4_a,c4_b,c4_c,c5_a,c5_b,c5_c,c6_a,c6_b,c6_c,c7_a,c7_b,c7_c,c8_a,c8_b,c8_c,c9_a,c9_b,c9_c, c10_a,c10_b,c10_c,c11_a,c11_b,c11_c,c12_a,c12_b,c12_c,c13_a,c13_b,c13_c,c14_a,c14_b,c14_c,c15_a,c15_b,c15_c,c16_a,c16_b,c16_c,c17_a,c17_b,c17_c,c18_a,c18_b,c18_c,c19_a,c19_b,c19_c, c20_a,c20_b,c20_c,c21_a,c21_b,c21_c,c22_a,c22_b,c22_c,c23_a,c23_b,c23_c,c24_a,c24_b,c24_c,c25_a,c25_b,c25_c,c26_a}{\ceps};
    \end{pgfonlayer}

\end{tikzpicture} 
    }
    \caption{A completed construction for the formula
        $\Phi = (x \vee y \vee z) \wedge
        (\overline{x} \vee y \vee \overline{z}) \wedge
        (\overline{x} \vee \overline{y} \vee z) \wedge
        (x \vee \overline{y} \vee \overline{z})$.}
    \label{fig:unique}
\end{figure}
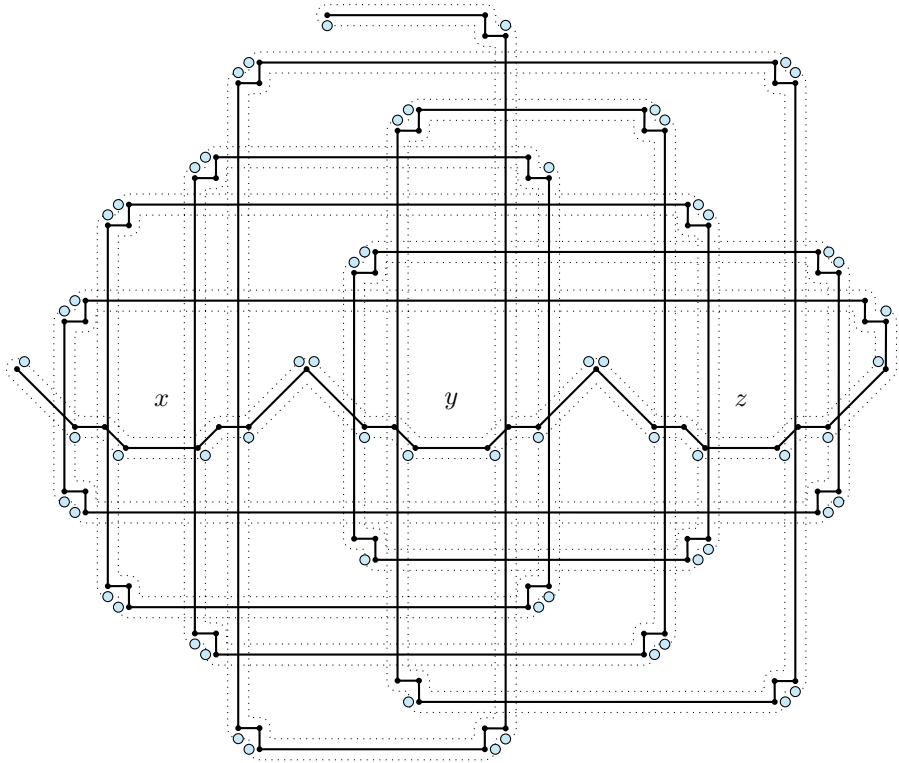

\subsection{Hardness Result}
\begin{lemma}
There exists a vertex-unique polygonal path $Q$ on $S$ with $\delta_F(P, Q) \le \eps$ if and only if the formula $\Phi$ is satisfiable.
\end{lemma}
\begin{proof}
For the forward direction, assume $\Phi$ has a satisfying assignment.  The variable portion of the construction always has a vertex-unique polygonal path $Q$ on $S$ with $\delta_F(P, Q) \le \eps$; assume each variable gadget is chosen according to the satisfying assignment.  Then each clause loop will have at least one literal point that can be used to change tracks before its dead end is reached.  After all clause loops have been traversed, the path will have \Frechet distance less than $\eps$ from P.

For the backward direction, assume $\Phi$ does not have a satisfying assignment.  Then, no matter how the initial variable portion of the construction is traversed, there will be at least one clause loop which will not be able to use any literal point it passes.  Once the end of the clause loop is reached, there will be no way to continue the path without increasing the \Frechet distance beyond $\eps$.
\end{proof}

The variable section contains two \acs for each variable, and the clause section contains six \acs for each clause, so the construction is clearly of polynomial size.  This leads to the following result.

\begin{theorem}
The Unique Subset Continuous CPSM Problem is NP-complete.
\end{theorem}

\section{Continuous Non-unique Subset CIPSM} \label{sec:imprecise}
In this section, we discuss the CIPSM problem, in which the goal is to resolve each region to a single point so that there exists a polygonal curve on those points with \Frechet distance at most $\eps$ from the given curve.  For simplicity, we will treat the imprecise points as line segments, but we observe that all results trivially extend to other regions.  The CIPSM problem has eight versions corresponding to the eight versions of the CPSM problem.  The following observation is immediate:

\begin{observation}
The five NP-complete versions of the CPSM problem imply the NP-completeness of their corresponding CIPSM problem versions.
\end{observation}

Here, we show that the Continuous Non-unique Subset CIPSM problem is NP-complete, using a reduction similar to the one presented in the previous section.  A key property of the construction in Section \ref{sec:precise} is that after the variable section has been traversed, exactly two of the four corner points of each variable remain usable by the clause section.  This is due to the fact that points cannot be reused, and two of the four points must be used to traverse the variable section.  To adapt the reduction to the Non-unique CIPSM problem, we simply connect the two points of each corner into a single imprecise segment (Figure \ref{fig:impcorners}).  Since each imprecise segment must resolve to a single point, and since points can be used more than once in this version of the problem, the end result after traversing the variable section is exactly the same: two points for each variable are available for the clause section to choose, either in the positive literal positions or the negative literal positions.  Thus, instead of ``using up'' corner points, we are ``making them available'' for the clause loops to use to escape their dead ends.

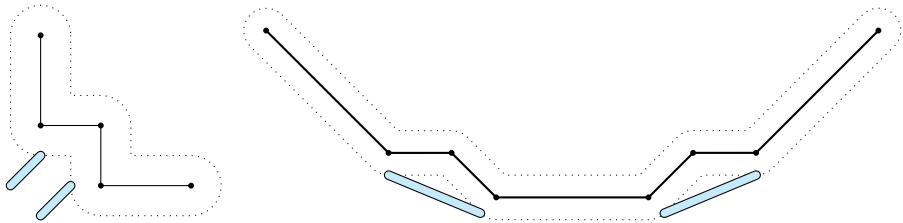
\begin{figure}
    \begin{subfigure}[t]{0.25\textwidth}
		\tikzsetnextfilename{imp_corner}
\begin{tikzpicture}[scale=.4,rotate=0]
    \tikzdefines

    \foreach [count=\x] \pt in {(0,5), (0,2), (2,2), (2,0), (5,0)}
        \node[cblack] (p\x) at \pt {};
    \draw (p1)--(p2)--(p3)--(p4)--(p5);
    \draw[dotted] \pathborder{p1,p2,p3,p4,p5}{1cm};

    %\node[cblue,label=0:$a$] (s1) at (1,5) {};
    \draw[double distance=.1cm,line cap=round,double=cyan!20] (-1,0) -- (0,1);
    %\node[cblue] (s2) at (0,1) {};
	\draw[double distance=.1cm,line cap=round,double=cyan!20] (0,-1) -- (1,0);
    %\node[cblue] (s3) at (1,0) {};
    %\node[cblue,label=90:$d$] (s4) at (5,1) {};

    %\draw[thick, black, dashed] (s1) -- (s2) -- (s4);
    %\draw[thick, black, dashed] (s1) -- (s3) -- (s4);
    %\draw[thick, red, dashed] (s1) -- (s2) -- (s3) -- (s4);

    %\draw[thick, red, dashed] (s1) -- (s2) -- (s4);
    %\draw[thick, red, dashed] (s1) -- (s3) -- (s4);
    %\draw[thick, red] (s1) -- (s2) -- (s3) -- (s4);

\end{tikzpicture} 
        \label{fig:impcorner}
    \end{subfigure}
    \hfill
    \begin{subfigure}[t]{0.75\textwidth}
        \tikzsetnextfilename{imp_variable}
\begin{tikzpicture}[scale=1.48]
    \tikzdefines

    \node[cblack] (a0_p) at (-0.5,1.5) {};
    \xcorner{a0}{0,1}{1,0}{1,0}{2,0};
    \xcorner{b0}{1,0}{2,0}{3.5,0}{4.5,1};
%        \node[cred] (a\x_q) at ($(a\x_p)!\ceps!90:(a\x_a)$) {};
%        \node[cred] (b\x_q) at ($(a\y_p)!\ceps!270:(b\x_c)$) {};
   \node[cblack] (a1_p) at (5,1.5) {};
   \draw[thick] (a0_p)--(a0_a)--(a0_b)--(a0_c)--(b0_a)--(b0_b)--(b0_c)--(a1_p);
   
   \draw[double distance=.1cm,line cap=round,double=cyan!20] (a0_d)--(a0_e);
   \draw[double distance=.1cm,line cap=round,double=cyan!20] (b0_d)--(b0_e);
   
%    \node () at ($($(a0_b)!.5!(b0_b)$)+(0,.5)$) {$x$};
%    \node () at ($($(a1_b)!.5!(b1_b)$)+(0,.5)$) {$y$};
%    \node () at ($($(a2_b)!.5!(b2_b)$)+(0,.5)$) {$z$};

    %\node[label={below:$x$}] () at (a0_d) {};
    %\node[label={below:$x$}] () at (b0_e) {};
    %\node[label={below:$\overline{x}$}] () at (a0_e) {};
    %\node[label={below:$\overline{x}$}] () at (b0_d) {};

    \begin{pgfonlayer}{background}
        \draw[black, dotted] \pathborder{a0_p,a0_a,a0_b,a0_c,b0_a,b0_b,b0_c,a1_p}{\ceps};
    \end{pgfonlayer}

\end{tikzpicture} 
		\label{fig:impvar}
    \end{subfigure}
    \caption{(Left) Most points are replaced with line segments that have one endpoint inside the cylinder.  (Right) The two corner points of the variable section's corners are replaced with a single segment joining them. \vspace{-10pt}}
        
     \label{fig:impcorners}
\end{figure}

We have modeled imprecise points as line segments for simplicity.  However, the model can easily be extended to disks by placing them so that they are tangent to the appropriate cylinders at the appropriate locations. Other shapes can also be positioned to correctly intersect the cylinders by aligning the cylinder boundaries at their extremal points.  Since the entire construction is scalable, there is no danger of being forced to place imprecise points close enough to interfere with each other. This leads to the following theorem.

\begin{theorem}
The Continuous Non-unique Subset CIPSM is NP-complete.
\end{theorem}

\section{Discrete CIPSM Problem} \label{sec:dimprecise}
In this section, we study the CIPSM under discrete \Frechet distance and deal with the remaining two open problems: Discrete Non-Unique versions, Subset and All-Points.

\subsection{Preliminiaries}
Discrete \Frechet distance \cite{Eiter94} is a variation of the standard \Frechet distance that only takes into account distance at the curve vertices.  For a given curve $P$ of length $n$, let $F(P)$ be the set of all finite sequences $a$ of length $m$ such that $a_1 = 1$, $a_{m} = n$, and $a_{i+1} \in \{a_i, a_i+1\}$ for all $i$.  Then the \emph{discrete \Frechet distance} between two curves $P$ and $Q$ can be defined as $\min_{a \in F(P), b \in F(Q)} \max_i d(P_{a_i}, Q_{b_i})$.

It is interesting to note that, since the edges of the given curve have no impact on the discrete \Frechet distance, and since we are allowed to visit the points of $S$ in any order, then the discrete \Frechet distance between a given realization of $S$ and $P$ is the same for any curve with the same vertex set as $P$.  Thus, the problem can be restated as follows: does there exist a realization of $S$ such that every $\eps$-ball around the vertices of $P$ contains at least one realized point?  Although the edges are irrelevant, we still include them in the figures below for the sake of clarity.  However, it is important to remember that the vertices can be connected in any way and the result will not change.  Thus, our reduction applies even if the input curve is non-intersecting.

\subsection{Reduction Outline}
In the reduction for the continuous version (Section \ref{sec:imprecise}), imprecise points overlapped with different cylinders to provide different options.  But when dealing with discrete \Frechet distance, cylinders don't matter; only the balls around the vertices do.  In this reduction, we will place imprecise points so that they overlap with consecutive vertices in the given curve, thus forcing a choice between resolving in the ball of the former vertex or that of the latter.

The key gadget in this reduction is shown in Figure \ref{fig:dimp_edge}.  Placing imprecise points so they overlap with every two consecutive vertices of the given curve allows information to be transfered along the length of the curve.  In the absence of point $a$, the first imprecise point must resolve to a point in the former of the two balls it overlaps.  This in turn forces the second imprecise point to resolve to a point in its former ball, and so on.  However, if $a$ is available for use, then the first imprecise point is free to resolve in the latter ball, as are all the rest, leaving the last imprecise point free to resolve to a point outside the current chain.  In other words, point $b$ is available if and only if point $a$ is available.  We refer to the configuration where $b$ is (not) available as the \emph{forward (backward) position}.

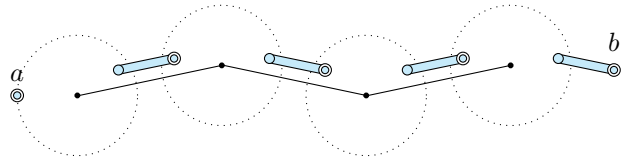
\begin{figure}[h]
    \centering
    \makebox[\textwidth][c]{
    \tikzsetnextfilename{dimp_edge}
\begin{tikzpicture}[scale=4]
    \tikzdefines
    \def\p{3}
    \foreach \x in {0, ..., \p} {
        \node[cblack] (p\x) at (\x*     \ceps*2.4 , {mod(\x, 2)/10}) {};
        \coordinate (t) at     ({(\x+1)*\ceps*2.4}, {mod((\x+1), 2)/10}) {};
        \begin{pgfonlayer}{background}
            \draw[dotted] (p\x) circle (\ceps);
        \end{pgfonlayer}
        \coordinate (qa\x) at ($(p\x)!.8*\ceps!20:(t)$);
        \coordinate (qb\x) at ($(t)!.8*\ceps!-20:(p\x)$);
        \draw[lblue] (qa\x)--(qb\x);
        \node[cblue] () at (qa\x) {};
        \node[cblue, double] () at (qb\x) {};
    };
    \foreach \x  [evaluate=\x as \y using (\x-1)] in {1, ..., \p} {
        \draw (p\y)--(p\x);
    }
    \node[cblue, double, label={above:$a$}] () at ($(p0) + (-\ceps, 0)$) {};
    \node[label={above:$b$}] () at (qb\p) {};
\end{tikzpicture} 
    }
    \caption{The imprecise points can resolve to the double-circle points if and only if $a$ exists.  Without $a$, the single-circle points are the only choice.}
    \label{fig:dimp_edge}
\end{figure}

While this strategy works well for transferring information, it is not sufficient by itself to create the necessary constructs to simulate variables.  To accomplish this type of true-or-false behavior, we can arrange the given curve in a circle, creating a circular dependency. This forces all imprecise points to resolve to either the former or latter of their respective balls.  Figure \ref{fig:dimp_variable_chain} shows three such constructs.

For the imprecise points in these variable constructs, the forward positions denote the value of {\sc true} for the variable and the backwards positions denote {\sc false}.  These points will then serve as the ``extra point'' in the information propagation gadgets discussed earlier, which will bring together those variables referenced in a specific clause. 

%\begin{figure}[h]
%    \centering
%    \makebox[\textwidth][c]{
%    \input{fig_dimp_variable.tex}
%    }
%    \caption{A variable gadget.}
%    \label{fig:dimp_variable}
%\end{figure}

\subsection{Full Construction}
Our construction must take as input a 3CNF formula $\Phi$ whose incidence graph (plus a simple cycle about the variable nodes) is planar.  Our full construction consists of three parts:
\begin{itemize}
\item The Variable section, which represents the variables of $\Phi$ and enforces a {\sc true} or {\sc false} value on each.
\item The Transfer section, which brings the three variables involved in each clause together, transferring information about the value of each variable.
\item The Clause section, which ensures that each clause contains at least one {\sc true} literal.
\end{itemize}

First, we describe the Variable section.  As described earlier, each variable is represented by a circular arrangement of imprecise points.  To ensure that edges can extend both upward and downward from the horizontal row of variable gadgets, the curve runs through top halves of each variable before doubling back and running through the bottom halves, as shown in Figure \ref{fig:dimp_variable_chain}.  The size of each cycle depends on how many times the variable is used.  The Transfer section then connects each variable to where the corresponding clause points will be.  The curve starts exactly $\eps$ distance away from the appropriate literal point and travels to the location the clause point will later be placed.  It then doubles back before moving onto the next literal, in order not to introduce any unnecessary edges that might inhibit future segments.  Finally, the Clause section visits each clause location without any extra imprecise points, forcing one of the edges connected to a variable to serve as the point to pair with.

\begin{figure}[h]
    \centering
    \makebox[\textwidth][c]{
    \tikzsetnextfilename{dimp_variable_chain}
\begin{tikzpicture}[scale=3]
    \tikzdefines
    \def\p{3}
    \def\rad{\p*\ceps*.7}
    \def\m{\ceps*6.5}
    \foreach \n in {0, ..., 5} { 
   	\begin{scope}[shift={({(abs(\n-2.5)-1)*\m},-0.05*floor(\n/3))}, rotate={180*floor(\n/3)}]
	    \foreach \x in {0, ..., \p} {
		    \node[cblack] (p\n\x) at (\x*180/\p:\rad) {};
		    \begin{pgfonlayer}{background}
		    	\draw[dotted] (p\n\x) circle (\ceps);
		    \end{pgfonlayer}
		};
	    \foreach \x  [evaluate=\x as \y using (\x-1)] in {1, ..., \p} {
	      	\coordinate (qa\n\x) at (\y*180/\p+\ceps*2:\rad*1.1) {};
	      	\coordinate (qb\n\x) at (\x*180/\p-\ceps*2:\rad*1.1) {};
	      	\draw[lblue] (qa\n\x)--(qb\n\x);
		 	\draw (p\n\y)--(p\n\x);
		};
	\end{scope}};
	%\draw (p0\p)--(p10);
    \foreach \x in {1, ..., 5} {
    	\pgfmathtruncatemacro{\y}{\x-1}%
    	\draw (p\y\p)--(p\x0);
	};
    \draw[arrows=>-] ($(p00)+(.3,0)$) --(p00);
	\draw[arrows=<-] ($(p5\p)+(.3,0)$) -- (p5\p);

	\node[cblack] (p1) at ($(qb11) + (45:\ceps)$) {}; 
	\draw[dotted] (p1) circle (\ceps);
	\node[cblack] (p2) at ($(p1) + (45:\ceps*2.1)$) {}; 
	\draw[dotted] (p2) circle (\ceps);
	\draw (p1)--(p2);
	
	\coordinate (p0) at ($(p1) + (135:\ceps*1.2)$);
	\draw (p0)--(p1);
	\draw[dashed] (p0)--($(p1)!2!(p0)$);

	\coordinate (p3) at ($(p2) + (0:\ceps*1.2)$);
	\draw (p2)--(p3);
	\draw[dashed] (p3)--($(p2)!2!(p3)$);
	   	
   	\draw[lblue] ($(p1)!.8*\ceps!20:(p2)$) -- ($(p2)!.8*\ceps!-20:(p1)$);

\end{tikzpicture} 
    }
    \caption{The Variable section.  One cycle is created per variable.}
    \label{fig:dimp_variable_chain}
\end{figure}
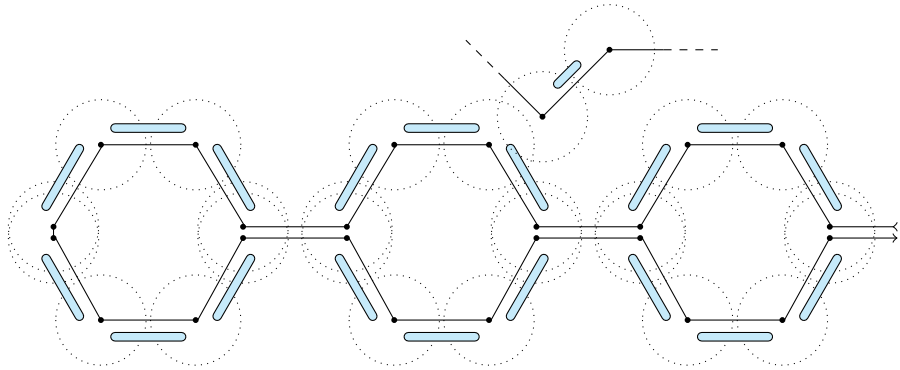

\subsection{NP-Completeness Proof}
\begin{lemma}
The construction has a valid curve if and only if $\Phi$ is satisfiable.
\end{lemma}
\begin{proof}
If $\Phi$ is satisfiable, then there is an assignment to the variables that will satisfy every clause.  Thus, every clause point in the construction will have a point to pair with, provided by one of the three edge constructs connected to it that link to a variable gadget.  Since the other parts of the construction always have at least one point to pair with under all circumstances, a valid curve exists.

Conversely, if $\Phi$ is not satisfiable, then there will be some clause point whose connected edges are forced to be in the backwards position, unable to provide the clause with a point to pair with.  In this case, no valid curve can exist.
\end{proof}

Furthermore, there are no unused points in our construction.  If a valid curve exists at all, then it is able to visit every point.  Thus, this proof applies to both the Subset and All-Points versions of the problem.

\begin{theorem}
The Discrete Non-Unique CIPSM problems, both Subset and All-Points, are NP-Complete, even when the given curve $P$ is simple.
\end{theorem}

\begin{table}[h]
    \centering
   
\begin{tabular}{l@{\hspace{1em}}l@{\hspace{3em}}|@{\hspace{2em}}c@{\hspace{2em}}c}
        \toprule
        \multicolumn{2}{l|@{\hspace{2em}}}{\textbf{CIPSM Problem}}     &  
Discrete & Continuous\\
        \midrule
        Subset  & Unique        & NP-C     & NP-C  \\
                & Non-Unique    & \textbf{NP-C}       & \textbf{NP-C} \\
        All-Points & Unique        & NP-C   & NP-C \\
                & Non-Unique    & \textbf{NP-C}   & NP-C \\
        \bottomrule
    \end{tabular}%
    \bigskip
    \caption{Eight versions of the CIPSM problem.  The bold entries are non-trivial to show. \vspace{-20pt}}
    \label{tab:impresults}%
\end{table}%

\section {An Approximation Algorithm for CPSM} \label{sec:approx}
We have now shown that the entire class of CIPSM problems are all NP-complete.  We note, however, that those versions that are in P for precise points have a simple approximation algorithm.  One can always take any realization of the imprecise points and apply the existing optimization version of the algorithm for the precise version of the problem, and the result will always be within an additive factor of the maximum diameter of all imprecise points in the given set.  We leave the development of a novel algorithm with a better approximation factor as an open problem.  However, in the remainder of this section, we present a non-trivial approximation algorithm for one of the NP-complete versions of the CPSM problem, the Continuous Non-Unique All-Points version.

We formulate a restricted version of the problem where points in $S$ are associated with segments in $P$ by proximity.  Forcing the curve construction over $S$ to respect this proximity yields an easier problem that is solvable in polynomial time using a modified version of the algorithm presented in \cite{Maheshwari11}.  Furthermore, its optimal solution has \Frechet distance at most 3 times that of the unrestricted optimal.  Algorithm 1, presented below, solves this restricted problem. \vspace{-10pt}

    \begin{algorithm}[h]
        \caption{Restricted Non-unique All-Points CPSM $(P, S, \eps)$}
        \label{algo:restricted}
        \begin{algorithmic}[1]
        %\Procedure{RestrictedCPSM}{$P, S, \eps$}
            \State Compute $S_i$ and $S^*_i$ for all $i$ \label{algo:closest}
            \State \textbf{If} any point is outside all $S_i$, \textbf{return} \textbf{no}  \label{algo:preproc}
            \State Compute $r_i(s, t)$ for all $1 \le i \le n$ and $s, t \in S$ \label{algo:precomp}
            \State Compute the entry and exit sets for each segment. \label{algo:entryexit}
            \State Modify $r_i(s, t)$ to obtain $r'_i(s, t)$ \label{algo:modify}
            \State Apply the Subset algorithm using $r'_i(s, t)$ \label{algo:subset}
            \State \textbf{Return} the result
        %\EndProcedure
        \end{algorithmic}
    \end{algorithm}
    \vspace{-10pt}
    
In this description, $S_i$ is the set of points inside the cylinder of the $i$th segment of $P$.  $S^*_i$ is a subset of $S_i$ for which the $i$th segment is the closest.  The algorithm in \cite{Maheshwari11} works by first computing reachability information in the form of a function $r_i(s, t)$.  We modify this reachability function for our purposes to obtain $r'_i(s, t)$, and then run the algorithm with this modified function.  Due to space limitations, we must omit further details, but a more complete description of the algorithm as well as proofs of correctness, time complexity, and approximation factor are given in the appendix.  Analysis of Algorithm 1 leads to the following theorem.

\begin{theorem}
The Continuous Non-Unique All-Points CPSM can be 3-approx-imated in $O(nk^2 \log(nk))$ time.
\end{theorem}

\bibliographystyle{plain}
\bibliography{frechet}

\pagebreak

\section{Appendix: CPSM Approximation Algorithm}
In this appendix, we consider the optimization version of the Continuous Non-Unique All-Points CPSM and detail the approximation algorithm for it that was briefly discussed in Section \ref{sec:approx}.  To do so, we develop an exact algorithm for a restricted version of the problem, which also serves as a 3-approximation to the unrestricted version.

The main combinatorial challenge of the All-Points version of the problem stems from the fact that a point in cylinders of multiple segments can be visited at any one of the segments.  To remove this challenge, we introduce an additional restriction to the problem; we enforce that each point in $S$ be visited at its closest segment.  Visiting points at other segments is also allowed, but each point must be visited at its closest segment even if it is also visited at another one.

In the unrestricted problem, the decision whether or not to visit a given point at a given segment must be made on the basis of whether or not it is possible to visit that point at some other segment.  The additional restriction disconnects this decision from the rest of the problem by making the decision dependent only on whether or not the point is closest to a given segment.  This makes the problem polynomial-time solvable.

\subsection{Preliminaries}
Let $P$ be a polygonal curve composed of $n$ segments in $\Real^d$, denoted by $(P_1, P_2, \dots, P_n)$.  We denote the cylinders $\cyl(P_i, \eps)$ as $C_i$.  For convenience, we define $C_0$ and $C_{n+1}$ to be the $\eps$-balls around the start and end point of $P$.  The set $S_i$ is defined as $S \cap C_i$.  Finally, for any point $s \in S$, we define $P_i[s]$ as the line segment $P_i \cap \ball(s, \eps)$.

Let $Q$ be a polygonal curve whose vertices are in $S$ and whose \Frechet distance from $P$ is at most $\eps$.  For some vertex $s \in S$, $Q$ is said to \emph{visit} a point $s \in S$ at segment $i$ if there exist subcurves $P'$ and $Q'$, each beginning the start of their respective curves, such that $Q'$ ends at $s$, $P'$ ends at some point $p \in P_i$, and $\delta_F(P', Q') \le \eps$.  A point $s \in S_i$ is said to be \emph{reachable} at $i$ if there exists a curve that visits it at $i$, and the pair $(s, p)$ is called a \emph{feasible pair}.

\subsection{Restricted Version Algorithm}
We follow the parametric search paradigm by first developing an algorithm for the decision version of the problem.  An obvious preprocessing step is to confirm that all points of $S$ are a member of some $S_i$.  Another is to confirm that $S_0$ and $S_{n+1}$ are non-empty.  If either of these conditions are false, we can stop and return ``{\sc false}'' immediately.  For a given $s \in S$, let $P^s$ be the segment of $P$ closest to $s$.  Let the \emph{essential points} of $P_i$, denoted by $S^*_i$, be the set $\{s \in S \mid P^s = P_i\}$.  Note that, under the preprocessing assumption, $S^*_i \subset S_i$.

Per our restriction, every point must be visited at its closest cylinder.  However, it may be necessary to visit points in other cylinders as well.  For example, even if $S^*_i = \emptyset$, some point $s \in S_i$ may need to be visited in order to stay close to the given curve and reach future points.  If we think of single points in multiple cylinders as if they were separate points, then there are two types: points we must visit, and points we may skip.  In this way, the problem can be thought of as a variant of the Subset version, in which all points are the latter type.

In order to visit every point in a segment's essential set, care must be taken regarding the first and last points visited for a given segment.  Let $s \in S_i$ be the first point visited in $P_i$, which may not be an essential point of $P_i$.  If there exists an essential point $s'$ for which $\pl(P_i[s]) \succ \pr(P_i[s'])$, then it will not be possible to visit $s'$; the curve has already gone too far and cannot backtrack far enough.  On the other hand, if $t \in S_i$ is the last point visited in $P_i$ and there exists an essential point $t'$ for which $\pl(P_i[t']) \succ \pr(P_i[t])$, then $t'$ must not have been visited, because it is too far ahead to have backtracked from.

To formalize this notion, we say a point $t$ is an \emph{entry point} for $P_i$ if $\pl(P_i[t]) \preceq \pr(P_i[s])$ for all $s \in S^*_i$.  Analogously, we say $t$ is an \emph{exit point} if $\pl(P_i[s]) \preceq \pr(P_i[t])$ for all $s \in S^*_i$.  Note that, if $S^*_i = \emptyset$, then every point in $S_i$ is an entry and exit point for that segment.  In order to ensure that every point in $S^*_i$ can be visited, we must enter each cylinder via an entry point and leave it through an exit point.  As long as this is enforced, we can simply enter each cylinder via an entry point, visit all the essential points in monotonic order along the segment, and exit through an exit point to the next cylinder.

\begin{lemma} \label{thm:entryexit}
Visiting every point at its closest segment is possible if and only if the first (last) point visited in each cylinder is an entry (exit) point for that segment.
\end{lemma}

To turn this lemma into an algorithm, we adapt the algorithm for the Subset version of the problem given in \cite{Maheshwari11}.  We provide a small review here.  The first step of the Subset algorithm is to precompute a reachability function $r_i(s, t)$.  Let $s \in S_i$ be a point that is reachable at $P_i$ by some feasible curve $Q$ ending in $s$.  Given a point $t \in S$, $r_i(s, t)$ is defined as the largest index $j \ge i$ such that the curve $Q + \overline{st}$ visits $t$ at $P_j$, or 0 if $Q + \overline{st}$ is not feasible.  As proven in \cite{Maheshwari11}, $t$ is reachable at $P_j$ for all $i \le j \le r_i(s, t)$. Therefore, this value provides reachability information for all pairs of points in $S$ from any segment to any other.  In order to ensure that no essential points be skipped, we must modify $r_i(s, t)$ to obtain a new function $r'_i(s, t)$ with the following properties:
\begin{itemize}
    \item $r'_i(s, t)$ must be either 0 or $i$ if $s$ is not an exit point for $P_i$.
    \item If $r'_i(s, t) > i$, then $t$ must be an entry point for $P_{r'_i(s, t)}$.
    \item All cylinders $j$ for $i < j < r'_i(s, t)$ must have empty essential sets.
\end{itemize}

Recall that every point in a cylinder with an empty essential set is an entry point.  Therefore, the previously stated property of $t$ being reachable at $C_j$ if $t \in C_j$ for $i \le j \le r'_i(s,t)$ still holds.

Under this modified reachability function, the Subset algorithm decides our restricted problem.  Note that, even though the actual curve returned by the Subset algorithm is not guaranteed to visit all points, it will return a curve that respects the modified reachability function.  By Lemma \ref{thm:entryexit}, this is sufficient to guarantee the existence of a valid all-points curve.

\textbf{Time Complexity.} Lines \ref{algo:closest} and \ref{algo:preproc} of Algorithm \ref{algo:restricted} takes $O(nk)$ time.
Lines \ref{algo:precomp} and \ref{algo:subset} take $O(nk^2)$ time \cite{Alt03,Maheshwari11}.
Computing the entry and exit sets on Line \ref{algo:entryexit} requires comparing $O(k)$ candidates with $O(k)$ other points, repeated for each of the $n$ cylinders, so this step takes $O(nk^2)$ time.
To compute $r'_i(s, t)$ in Line \ref{algo:modify}, we define a value $e_i$ as the first segment after $i$ with a non-empty essential set, which can be computed in $O(n)$ time.  Then, given $r_i(s, t)$, we set $r'_i(s, t)$ to be the smaller of $e_i$ and $r_i(s, t)$.  Since correcting a single entry takes $O(1)$ time, it takes $O(nk^2)$ to correct the entire function.
Thus, the complexity of the algorithm is $O(nk^2)$.

\begin{theorem}
Algorithm 1 correctly decides the restricted version of the Continuous Non-unique Subset CPSM in $O(nk^2)$ time.
\end{theorem}

With an algorithm for the decision version in hand, the technique of parametric search is employed to find the optimal curve.  By analyzing the so-called \emph{free space diagram} of $P$ and each of the $S \times S$ possible segments of $Q$, $O(nk^2)$ critical values of $\eps$ can be identified.  These values can then be sorted, and the decision version of the algorithm can be used to binary search for the smallest value.  This technique yields an algorithm with running time $O(nk^2 \log(nk))$.

\subsection{Approximation Proof}
We now show that the optimal curve for the unrestricted problem can be transformed into a curve that obeys the extra restriction while only increasing its \Frechet distance by a factor of 3.  This will show that the restricted version algorithm is a 3-factor approximation algorithm for the unrestricted version.

For a given $P$ and $S$, let $Q$ be the optimal curve for the unrestricted problem, with \Frechet distance $\eps$ from $P$.  Some points of $S$ may not be visited at their closest segment; let $s$ be such a point.  It must be true that $s$ is within $\eps$ of $P^s$, or a curve with \Frechet distance $\eps$ would not be possible.  Recall that since $Q$ is within \Frechet distance $\eps$ of $P$, there is always at least one feasible pair for any point on $P$.  Let $(s', p)$ be a feasible pair such that $s \in \ball(p, \eps)$.  Then, add $s'$ as a new vertex of $Q$.  Note that the distance between $s$ and $s'$ is at most $2\eps$.  Repeating this process for every point not visited at its closest segment yields a new curve $Q'$.  Since each new vertex has been added along an existing segment, $\delta_F(P, Q) = \delta_F(P, Q')$.

Now, merge each $s'$ with its corresponding $s$ by translating the former to the position of the latter, yielding a new curve $Q''$ with a potentially different \Frechet distance from $P$.  Let $\sigma$ and $\tau$ be reparameterizations of $P$ and $Q'$, and consider the point $Q'(\tau(t))$ for some $t \in [0, 1]$, which lies on some segment of $Q'$.  The endpoints of the corresponding segment in $Q''$ may have been perturbed up to $2\eps$, and thus the point $Q''(\tau(t))$ may be up to $2\eps$ away from $Q'(\tau(t))$.  Therefore, $\norm{P(\sigma(t)), Q''(\tau(t))}$ can be at most $2\eps$ larger than $\norm{P(\sigma(t)), Q'(\tau(t))}$.  Finally, since the \Frechet distance is the infimum of the maximum distance over all reparameterizations, we have that $\delta_F(P, Q'') \le \delta_F(P, Q') + 2\eps = 3\eps$.  This implies that the \Frechet distance between $P$ and the optimal restricted path is at most 3 times that of the unrestricted path.

\begin{theorem}
Given a polygonal curve $P$ and a point set $S$, a polygonal curve $Q$ whose vertices are exactly $S$ with $\delta_F(P, Q)$ at most 3 times that of the optimal can be computed in $O(nk^2 \log(nk))$ time.
\end{theorem}

As Figure \ref{fig:tightness} shows, the approximation bound is realizable.  However, if the algorithm yields a \Frechet distance for which no point in $S$ belongs to more than one cylinder, then this solution must also be optimal for the unrestricted version.
 
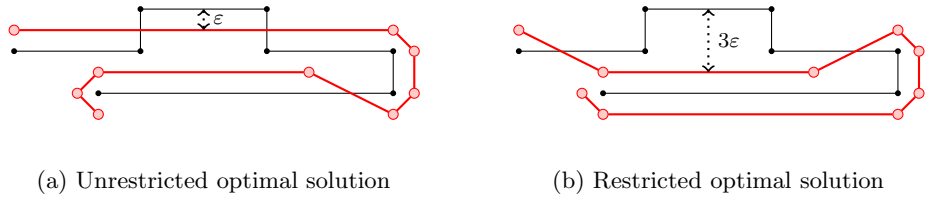
\begin{figure}[h]
    \begin{subfigure}[t]{0.45\textwidth}
        \tikzsetnextfilename{tightness2}
\begin{tikzpicture}[scale=.28]
    \tikzdefines
    \renewcommand{\ceps}{1cm}

    \node[cblack](p1) at (2,0) {};
    \node[cblack](p2) at (8,0) {};
    \node[cblack](p3) at (8,2) {};
    \node[cblack](p4) at (14,2) {};
    \node[cblack](p5) at (14,0) {};
    \node[cblack](p6) at (20,0) {};
    \node[cblack](p7) at (20,-2) {};
    \node[cblack](p8) at (6,-2) {};

    \node[cred](s1) at (2,1) {};
    \node[cred](s2) at (6,-1) {};
    \node[cred](s3) at (16,-1) {};
    \node[cred](s4) at (20,1) {};
    \node[cred](s5) at (21,0) {};
    \node[cred](s6) at (21,-2) {};
    \node[cred](s7) at (20,-3) {};
    \node[cred](s8) at (5,-2) {};
    \node[cred](s9) at (6,-3) {};

    \draw (p1)--(p2)--(p3)--(p4)--(p5)--(p6)--(p7)--(p8);

    \draw[red,thick] (s1)--(s4)--(s5)--(s6)--(s7)--(s3)--(s2)--(s8)--(s9);
    %\draw[red,thick] (s1)--(s4)--(s5)--(s3)--(s2)--(s6);

    %\path[draw,black,thick,dotted,<->] ($(p3)!.5!(p4)$) -- node[auto] {{\fontsize{8}{10}$3\varepsilon$} } ($(s2)!.5!(s3)$);

    \path[draw,black,thick,dotted,<->] ($(p3)!.5!(p4)$) -- node[auto] {{\fontsize{8}{10}$\varepsilon$} } ($(s1)!.5!(s4)$);
\end{tikzpicture}
        %\bigskip
        \caption{Unrestricted optimal solution}
        \label{fig:tightness1}
    \end{subfigure}
    \hfill
    \begin{subfigure}[t]{0.45\textwidth}
        \tikzsetnextfilename{tightness}
\begin{tikzpicture}[scale=.28]
    \tikzdefines
    \renewcommand{\ceps}{1cm}

    \node[cblack](p1) at (2,0) {};
    \node[cblack](p2) at (8,0) {};
    \node[cblack](p3) at (8,2) {};
    \node[cblack](p4) at (14,2) {};
    \node[cblack](p5) at (14,0) {};
    \node[cblack](p6) at (20,0) {};
    \node[cblack](p7) at (20,-2) {};
    \node[cblack](p8) at (6,-2) {};

    \node[cred](s1) at (2,1) {};
    \node[cred](s2) at (6,-1) {};
    \node[cred](s3) at (16,-1) {};
    \node[cred](s4) at (20,1) {};
    \node[cred](s5) at (21,0) {};
    \node[cred](s6) at (21,-2) {};
    \node[cred](s7) at (20,-3) {};
    \node[cred](s8) at (5,-2) {};
    \node[cred](s9) at (6,-3) {};
    
    \draw (p1)--(p2)--(p3)--(p4)--(p5)--(p6)--(p7)--(p8);

    \draw[red,thick] (s1)--(s2)--(s3)--(s4)--(s5)--(s6)--(s7)--(s9)--(s8);
    %\draw[red,thick] (s1)--(s4)--(s5)--(s3)--(s2)--(s6);

    \path[draw,black,thick,dotted,<->] ($(p3)!.5!(p4)$) -- node[auto] {{\fontsize{8}{10}$3\varepsilon$} } ($(s2)!.5!(s3)$);

    %\path[draw,black,thick,dotted,<->] ($(p3)!.5!(p4)$) -- node[auto] {{\fontsize{8}{10}$\varepsilon$} } ($(s1)!.5!(s4)$);
\end{tikzpicture}
        \caption{Restricted optimal solution}
        \label{fig:tightness2}
    \end{subfigure}
     \caption{If the two middle points are slightly closer to the top segments than the bottom segment, the optimal solution for the restricted version has \Frechet distance 3 times that of the unrestricted version.}
     \label{fig:tightness}
\end{figure}

\end{document}